\documentclass[a4paper,11pt,nopdfoutputerror]{quantumarticle}
\usepackage{amsmath,amssymb,amsthm}
\usepackage{booktabs}
\usepackage{array}
\usepackage[colorlinks=true,linkcolor=blue!60!black,citecolor=blue!60!black,urlcolor=blue!60!black]{hyperref}

\newtheorem{theorem}{Theorem}
\newtheorem{lemma}[theorem]{Lemma}
\newtheorem{proposition}[theorem]{Proposition}
\theoremstyle{remark}
\newtheorem*{remark}{Remark}

\newcommand{\Tr}{\mathrm{Tr}}
\newcommand{\NR}{\mathcal{N}_{R}}
\newcommand{\NS}{\mathcal{N}_{\bar S}}
\newcommand{\Smin}{S_p^{\min}}
\newcommand{\PT}{\Gamma}

\title{Exact certification of a positive-order R\'enyi additivity violation for an explicit channel pair}
\author{Artus Krohn-Grimberghe}
\affiliation{Percivio Ltd.}
\date{August 16, 2026}

\begin{document}

\maketitle

\begin{abstract}
Cubitt, Harrow, Leung, Montanaro, and Winter (CHLMW) exhibited an
explicit pair of quantum channels whose minimum output R\'enyi entropy
is nonadditive at order zero, and reported numerical violations at
positive orders close to zero~\cite{chlmw}. Their paper states that a
PPT-relaxed semidefinite-programming argument yields a rigorous
positive-order bound of order $10^{-2}$ for the same printed supports, without
printing the concrete endpoint, a dual witness, or a verifiable
certificate; a recent paper by Leung, Lovitz, and Wu summarizes the
public record for this pair as ``no rigorous endpoint was
obtained''~\cite{llw}. We close that gap for the trace-preserving
normalization of the printed pair fixed in
Section~\ref{sec:channels}. Small rational witness matrices prove that
every output of either channel has all eigenvalues between
$301/100000$ and $2/3$; a single explicit entangled input has an exact
rational joint output spectrum of rank eight; and two independent
elementary interval arguments turn these three facts into a proof of
strict additivity violation,
\[
  \Smin(\NR\otimes\NS) \;<\; \Smin(\NR)+\Smin(\NS),
\]
for every real order $0<p\le 1/22$. Every step of the verification
reduces to comparisons of integers, and the complete certificate is a
few small rational matrices that a reader can check with a short
program --- or, for any single order, by hand. To our knowledge,
consistent with the assessment of~\cite{llw}, this is the first
printed, computer-verifiable certified positive-order endpoint for this
explicit pair. We claim no novelty for the phenomenon or for the
eigenvalue-floor mechanism, both due to CHLMW, no priority over their
unprinted calculation, and no optimality of the endpoint.
\end{abstract}

\section{Introduction}\label{sec:intro}

For a quantum state $\rho$ and a real order $0<p<1$, the R\'enyi
entropy is
\[
  S_p(\rho)=\frac{\ln \Tr\rho^p}{1-p},
\]
and for a quantum channel $\mathcal{N}$ the minimum output R\'enyi
entropy is $\Smin(\mathcal{N})=\min_\rho S_p(\mathcal{N}(\rho))$, the
minimum over all input states. Whether this quantity is additive under
tensor products, $\Smin(\mathcal{N}_1\otimes\mathcal{N}_2)
=\Smin(\mathcal{N}_1)+\Smin(\mathcal{N}_2)$, was for a decade a central
question of quantum information theory, and its failure is by now known
for most orders. Werner and Holevo gave the first explicit
counterexample, at orders $p>4.79$~\cite{wernerholevo}; constructive
families followed for every $p>2$~\cite{grudka,szczygielski} and,
recently, for every $p>1$~\cite{derksen-lovitz}. Random constructions
settle every $p>1$ (Hayden and Winter~\cite{haydenwinter}) and the von
Neumann point $p=1$ (Hastings~\cite{hastings}), whose neighborhood
follows by continuity. The regime of small positive $p$ was settled
\emph{in existence} by CHLMW~\cite{chlmw}: they printed an explicit
pair of subspaces of $\mathbb{C}^4\otimes\mathbb{C}^3$ whose associated
channels violate additivity at $p=0$ (a rank-counting statement), they
reported, without proof, numerical violations for positive $p$ up to
roughly $0.1$ (in two normalization variants; see
Section~\ref{sec:channels}), and their Section~IV asserts that a
PPT-relaxed semidefinite-programming argument yields a rigorous bound
of order $10^{-2}$ for the same printed supports, without printing the concrete
endpoint, an optimizer, or a certificate. Leung, Lovitz, and
Wu~\cite{llw} prove violations for all $p>3/4$ and all $0\le p<1/4$ by
probabilistic (random-projection) methods, which do not print concrete
channel matrices together with a checkable certificate; an early claim
of explicit intervals covering $(0,0.2855)$ and $(0.7145,1)$ was
withdrawn by its authors for a crucial error~\cite{yuying}. About the
explicit CHLMW pair at positive order, reference~\cite{llw} states:
``Numerical evidence suggested that this pair continues to violate
additivity up to approximately $p=0.11$, but no rigorous endpoint was
obtained'' --- a statement about the public record, consistent with
CHLMW's unprinted SDP assertion. (The $\approx 0.1$-scale numerical
figures refer to normalization variants different from the channels
certified here; see Section~\ref{sec:channels}.)

Three statements must therefore be kept apart. (i)~That the pair
violates additivity on \emph{some} positive interval is CHLMW's: the
violation at $p=0$ plus continuity of the R\'enyi entropy in $p$
already proves the existence of an unspecified $p_0>0$. (ii)~A rigorous
bound of order $10^{-2}$ from the PPT-relaxed SDP is \emph{asserted} by
CHLMW but never printed. (iii)~What has been missing from the public
record --- and what this note supplies --- is a concrete endpoint
together with a checkable certificate. We do not improve any
existential range --- our interval $(0,1/22]$ lies inside the
existential range $0\le p<1/4$ of~\cite{llw}; ours is a printed
explicit pair with a checkable interval inside $(0,1)$.

This note supplies a rigorous endpoint. The proof is deliberately
elementary and certificate-based: beyond linear algebra at the level of
a first course, the only analytic facts used are that $t\mapsto t^p$ is
concave on $[0,\infty)$ for $0<p<1$ with derivative $p\,t^{p-1}$, and
that $t\mapsto x^t$ is decreasing for $0<x<1$. Every inequality that a
computer verifies is an inequality between two integers. The
certificate data consists of four rational witness matrices (two per
channel), one rational spectrum, and a short list of rational grid
points; all of it is printed in or shipped with this paper.

\begin{theorem}\label{thm:main}
Let $\NR$ and $\NS$ be the channels from $\mathbb{C}^4$ to
$\mathbb{C}^3$ obtained from the printed CHLMW subspaces $R$ and
$\bar S$ (the entrywise conjugate of $S$) by the inverse-input-marginal
trace-preserving normalization of Section~\ref{sec:channels}. Then for
every real $p$ with $0<p\le 1/22$,
\[
  \Smin(\NR\otimes\NS)\;<\;\Smin(\NR)+\Smin(\NS).
\]
\end{theorem}

The theorem does not identify the true minimizing inputs of either
channel, and it does not identify the largest order at which the pair
violates additivity; Section~\ref{sec:interval} explains why the
present method cannot be pushed far beyond $1/22$ without new
witnesses.

\paragraph{Claim boundary.} The explicit supports, the order-zero
construction and violation, and the idea of proving a positive-order
violation through a uniform output-eigenvalue bound are all due to
CHLMW~\cite{chlmw}; the particular CPTP representatives certified here
are the ones fixed in Section~\ref{sec:channels}. Our
contribution is limited to explicit constants, printed certificates,
and independently repeatable verification: the first \emph{printed},
checkable positive-order endpoint for this pair, in the sense
of~\cite{llw} quoted above; we claim priority only for the printed,
checkable object, not over any unpublished calculation. The endpoint
$1/22$ is not claimed to be optimal. The witness matrices were found by an AI-assisted numerical
search; the search plays no role in the proof, and
Section~\ref{sec:verification} describes what was verified and how.

\section{The channel pair}\label{sec:channels}

CHLMW print two six-dimensional orthogonal subspaces
$R,S\subset\mathbb{C}^4\otimes\mathbb{C}^3$; their twelve (unnormalized)
basis vectors are reproduced in Appendix~\ref{app:subspaces}, and every
entry is $0$, $\pm 1$, or a cube root of unity. Write $A=\mathbb{C}^4$
for the input system and $B=\mathbb{C}^3$ for the output system, let
$P_T$ be the orthogonal projection onto a subspace $T$, and let
$\bar S$ denote the entrywise complex conjugate of $S$.

A subspace projection is turned into a channel by the standard
Choi correspondence. Because we work at positive order, the
normalization must be fixed explicitly --- at $p=0$ only ranks matter
and any channel with Choi support $T$ gives the same statement, but the
positive output eigenvalues are not invariant under changes of
normalization. We use the canonical (inverse-input-marginal) choice ---
the standard trace-preserving normalization, of the same
local-filtering type as adopted in~\cite{llw}. For that reason, ``the CHLMW pair'' in
this note always means the printed supports together with this
normalization. CHLMW themselves report numerical violations up to
$p\approx 0.096$ for a normalized-projection variant and up to
$p\approx 0.112$ for a numerically weighted variant on the same
supports; those are different normalizations, and neither number is
claimed to transfer to the channels defined here. The subspace
construction itself is a variation of the completely entangled
subspaces of Cubitt, Montanaro, and Winter~\cite{cmw}.
For $T\in\{R,\bar S\}$, exact computation from the printed vectors
gives the input marginals
\begin{align*}
  M_R&=\Tr_B(P_R)=\mathrm{diag}\!\left(\tfrac53,\tfrac32,\tfrac32,\tfrac43\right),\\
  M_{\bar S}&=\Tr_B(P_{\bar S})=\mathrm{diag}\!\left(\tfrac43,\tfrac32,\tfrac32,\tfrac53\right),
\end{align*}
both invertible, and we define
\begin{align*}
  J_T&=\bigl(M_T^{-1/2}\otimes I_B\bigr)\,P_T\,\bigl(M_T^{-1/2}\otimes I_B\bigr),\\
  \mathcal{N}_T(X)&=\Tr_A\!\left[J_T\,(X^{\mathsf T}\otimes I_B)\right].
\end{align*}
Then $J_T\succeq 0$ and $\Tr_B(J_T)=I_A$ exactly, so $\mathcal{N}_T$ is
completely positive and trace preserving. The single computation that
connects $J_T$ to channel outputs is the following.

\begin{lemma}[Product expectations are output matrix elements]\label{lem:product}
For unit vectors $x\in A$ and $y\in B$,
\[
  \langle x\otimes y|\,J_T\,|x\otimes y\rangle
  \;=\;
  \langle y|\,\mathcal{N}_T(\bar x\bar x^{*})\,|y\rangle ,
\]
where $\bar x$ is the entrywise conjugate of $x$.
\end{lemma}

\begin{proof}
Write $J$ for $J_T$ with entries $J_{(a,b),(a',b')}$. By the definition
of $\mathcal{N}_T$,
$\langle y|\mathcal{N}_T(\bar x\bar x^{*})|y\rangle
=\sum_{a,a',b,b'} J_{(a,b),(a',b')}\,
(\bar x\bar x^{*})^{\mathsf T}_{a'a}\,\bar y_b\, y_{b'}$.
Now $(\bar x\bar x^{*})^{\mathsf T}_{a'a}=(\bar x\bar x^{*})_{aa'}
=\bar x_a x_{a'}$, so the sum equals
$\sum \bar x_a \bar y_b\, J_{(a,b),(a',b')}\, x_{a'} y_{b'}
=\langle x\otimes y|J|x\otimes y\rangle$.
\end{proof}

Since conjugation $x\mapsto\bar x$ maps the set of unit vectors onto
itself, controlling all product expectations of $J_T$ controls all
matrix elements $\langle y|\mathcal{N}_T(\phi\phi^*)|y\rangle$ over all
pure inputs $\phi$ and all unit $y$.

\section{Eigenvalue floor and cap}\label{sec:floorcap}

The certificates for the floor and the cap use one two-line mechanism.
For a matrix $Q$ on $A\otimes B$, the partial transpose $Q^{\PT}$ (on
$B$) has entries $(Q^{\PT})_{(a,b),(a',b')}=Q_{(a,b'),(a',b)}$.

\begin{lemma}[Decompositions bound product expectations]\label{lem:decomp}
Let $H$ be Hermitian on $A\otimes B$ and suppose
$H=P+Q^{\PT}$ with $P\succeq\alpha I$ and $Q\succeq\beta I$ for real
$\alpha,\beta$. Then for all unit vectors $u\in A$, $v\in B$,
\[
  \langle u\otimes v|\,H\,|u\otimes v\rangle\;\ge\;\alpha+\beta .
\]
\end{lemma}

\begin{proof}
For any $u,v$,
$\langle u\otimes v|Q^{\PT}|u\otimes v\rangle
=\sum \bar u_a\bar v_b\,Q_{(a,b'),(a',b)}\,u_{a'}v_{b'}
=\langle u\otimes\bar v|Q|u\otimes\bar v\rangle$,
by relabeling $b\leftrightarrow b'$. Hence
$\langle u\otimes v|H|u\otimes v\rangle
=\langle u\otimes v|P|u\otimes v\rangle
+\langle u\otimes\bar v|Q|u\otimes\bar v\rangle
\ge\alpha+\beta$.
\end{proof}

The shipped certificate stores, for each $T\in\{R,\bar S\}$, two
rational Hermitian matrices: a floor witness, recorded together with
the rational scale $\tfrac{17}{16}$ that the verifier applies before
checking, giving $Q_T^{\mathrm{low}}$, and a cap witness
$Q_T^{\mathrm{cap}}$. The matrices $P_T^{\mathrm{low}}$ and
$P_T^{\mathrm{cap}}$ are not stored; they are defined by subtraction
from the first and third identities below, and their entries are in
general algebraic rather than rational because $J_T$ is. (The floor
file also carries historical metadata --- $\delta=1/400$, $p=1/100$
--- from the project that produced it; the exact number-field program
of Section~\ref{sec:verification} certifies the floor in that
historical convention, while the interval-arithmetic program certifies
the constants displayed here.) The four matrices so defined satisfy
\emph{exactly}
\begin{gather*}
  J_T=P_T^{\mathrm{low}}+\bigl(Q_T^{\mathrm{low}}\bigr)^{\PT},\\
  P_T^{\mathrm{low}}\succeq\tfrac{3}{1000}I,
  \qquad Q_T^{\mathrm{low}}\succeq\tfrac{1}{100000}I,\\[0.7ex]
  \tfrac23 I-J_T=P_T^{\mathrm{cap}}+\bigl(Q_T^{\mathrm{cap}}\bigr)^{\PT},\\
  P_T^{\mathrm{cap}}\succeq\tfrac{1}{500}I,
  \qquad Q_T^{\mathrm{cap}}\succeq\tfrac{1}{500}I.
\end{gather*}
All eight positivity statements are verified in exact arithmetic
(Section~\ref{sec:verification}); the constants contain visible slack
and are not numerical tolerances.

\begin{proposition}[Floor and cap]\label{prop:floorcap}
For every input state $\rho$ and $T\in\{R,\bar S\}$, every eigenvalue
$\lambda$ of $\mathcal{N}_T(\rho)$ satisfies
\[
  \ell\le\lambda\le u,
  \quad
  \ell=\tfrac{3}{1000}+\tfrac{1}{100000}=\tfrac{301}{100000},
  \quad
  u=\tfrac23 .
\]
\end{proposition}

\begin{proof}
By Lemmas~\ref{lem:product} and~\ref{lem:decomp}, for every pure input
$\phi$ and unit $y$ we have
$\langle y|\mathcal{N}_T(\phi\phi^*)|y\rangle\ge\ell$ (apply the lemmas
with $H=J_T$) and
$\tfrac23-\langle y|\mathcal{N}_T(\phi\phi^*)|y\rangle\ge\tfrac{1}{500}+\tfrac{1}{500}>0$
(apply them with $H=\tfrac23 I-J_T$). A Hermitian matrix whose
expectations in all unit vectors lie in $[\ell,u]$ has all eigenvalues
in $[\ell,u]$. Mixed inputs are convex combinations of pure ones, and
$\mathcal{N}_T$ is linear, so the same bounds hold for all inputs.
\end{proof}

\section{From the floor and cap to a single-channel bound}\label{sec:single}

Fix $0<p<1$ and write $f(t)=t^p$, which is concave on $[0,\infty)$ with
$f'(t)=p\,t^{p-1}$ decreasing. Every output spectrum
$\lambda_1\ge\lambda_2\ge\lambda_3$ of either channel satisfies
$\lambda_1\le u$, $\lambda_3\ge\ell$, and
$\lambda_1+\lambda_2+\lambda_3=1$. The extreme admissible spectrum is
\begin{align*}
  e=(e_1,e_2,e_3)&=\bigl(u,\,1-u-\ell,\,\ell\bigr)\\
   &=\Bigl(\tfrac23,\,\tfrac{99097}{300000},\,\tfrac{301}{100000}\Bigr).
\end{align*}

\begin{lemma}[Envelope]\label{lem:envelope}
For every admissible spectrum $\lambda$ as above,
\[
  f(\lambda_1)+f(\lambda_2)+f(\lambda_3)\;\ge\;f(e_1)+f(e_2)+f(e_3).
\]
\end{lemma}

\begin{proof}
Since $f$ is concave, its tangent line at any point lies above its
graph: $f(t)\le f(s)+f'(s)(t-s)$ for all $s>0$, $t\ge0$. Taking
$s=\lambda_i$ and $t=e_i$ gives
\begin{equation}\label{eq:tangent}
  f(\lambda_i)-f(e_i)\;\ge\; c_i\,(\lambda_i-e_i),
\end{equation}
where $c_i=f'(\lambda_i)=p\,\lambda_i^{\,p-1}$.
Because $\lambda_1\ge\lambda_2\ge\lambda_3>0$ and $f'$ is decreasing,
$c_1\le c_2\le c_3$. Write the partial sums
$s_k=\lambda_1+\dots+\lambda_k$ and $t_k=e_1+\dots+e_k$. The
constraints give $s_1\le t_1$ (that is, $\lambda_1\le u$),
$s_2=1-\lambda_3\le 1-\ell=t_2$, and $s_3=t_3=1$. Summing
\eqref{eq:tangent} over $i$ and regrouping the right-hand side by
partial summation,
\begin{multline*}
  \sum_{i=1}^{3}c_i(\lambda_i-e_i)
  =(c_1-c_2)(s_1-t_1)\\
  +(c_2-c_3)(s_2-t_2)+c_3(s_3-t_3)\;\ge\;0,
\end{multline*}
since in each of the first two products both factors are $\le 0$ and
the last term vanishes.
\end{proof}

Define
\begin{equation}\label{eq:A}
  A(p)=\Bigl(\tfrac23\Bigr)^{p}
      +\Bigl(\tfrac{99097}{300000}\Bigr)^{p}
      +\Bigl(\tfrac{301}{100000}\Bigr)^{p}.
\end{equation}
Combining Proposition~\ref{prop:floorcap} and
Lemma~\ref{lem:envelope}, for every input $\rho$ and every
$T\in\{R,\bar S\}$ we get
$\Tr\bigl(\mathcal{N}_T(\rho)^p\bigr)\ge A(p)$,
and therefore, since $\ln$ is increasing and $1-p>0$,
\begin{equation}\label{eq:Smin-lower}
  \Smin(\mathcal{N}_T)\;\ge\;\frac{\ln A(p)}{1-p}.
\end{equation}

\section{One explicit joint input}\label{sec:joint}

Let $w_R=(\tfrac53,\tfrac32,\tfrac32,\tfrac43)$ and
$w_{\bar S}=(\tfrac43,\tfrac32,\tfrac32,\tfrac53)$ be the diagonal
entries of the two input marginals, and let $\Psi\in A\otimes A$ be the
unit vector proportional to
\[
  \sum_{j=1}^{4}\sqrt{(w_R)_j\,(w_{\bar S})_j}\;|j\rangle\otimes|j\rangle .
\]
(The weights are the ones induced by the two inverse-marginal
normalizations; the unweighted maximally entangled vector does not
produce the rational spectrum below.) Exact computation gives the joint
output
$\sigma=(\NR\otimes\NS)\bigl(\Psi\Psi^{*}\bigr)$ with the exact spectrum
\[
  \mathrm{spec}(\sigma)=
  \bigl\{0,
  \tfrac{5}{46},\tfrac{5}{46},
  \tfrac{18}{161},\tfrac{18}{161},
  \tfrac{51}{322},\tfrac{51}{322},
  \tfrac{39}{322},\tfrac{39}{322}\bigr\},
\]
a rank-eight rational spectrum on the nine-dimensional output space
(the entries sum to one). Define
\begin{equation}\label{eq:B}
  B(p)=2\Bigl(\tfrac{5}{46}\Bigr)^{p}
      +2\Bigl(\tfrac{18}{161}\Bigr)^{p}
      +2\Bigl(\tfrac{51}{322}\Bigr)^{p}
      +2\Bigl(\tfrac{39}{322}\Bigr)^{p}.
\end{equation}
Since $\sigma$ is the output of one particular input,
$\Tr\sigma^p=B(p)$ gives
\begin{equation}\label{eq:Smin-upper}
  \Smin(\NR\otimes\NS)\;\le\;\frac{\ln B(p)}{1-p}.
\end{equation}

By \eqref{eq:Smin-lower} and \eqref{eq:Smin-upper},
Theorem~\ref{thm:main} follows once we prove
\begin{equation}\label{eq:master}
  A(p)^2\;>\;B(p)
  \qquad\text{for all } 0<p\le\tfrac{1}{22}.
\end{equation}
At $p=0$ this inequality degenerates exactly to CHLMW's rank count:
$A(0)^2=3^2=9$ output dimensions against $B(0)=8$ nonzero joint
eigenvalues. The content of the present note is that the printed
witnesses keep the inequality strict on a whole interval.

\section{Two elementary interval arguments}\label{sec:interval}

Inequality \eqref{eq:master} involves the transcendental functions
$x^p$, but on the certified interval it reduces to finitely many
comparisons of integers. Both arguments below rest on the same
high-school kernel: \emph{a claimed bound $m/n$ on a rational power is
an integer inequality.} For instance,
\[
  \Bigl(\tfrac23\Bigr)^{1/22}>\tfrac{m}{n}
  \quad\Longleftrightarrow\quad
  2\,n^{22}>3\,m^{22},
\]
and the right-hand side is checked by integer multiplication alone. The
verifier encloses every needed power $x^{a/b}$ between two nearby
rationals in exactly this way (by binary search on a dyadic grid), and
all roundings are outward, so no comparison is ever decided by an
approximation.

\subsection{The staircase argument}\label{sec:staircase}

All seven distinct bases appearing in $A(p)$ and $B(p)$ lie strictly
between $0$ and $1$, and for $0<x<1$ the function $p\mapsto x^p$ is decreasing.
Hence on any cell $[l,r]\subset[0,1]$,
\[
  A(p)\ge A(r)
  \quad\text{and}\quad
  B(p)\le B(l)
  \quad(l\le p\le r).
\]
The certificate partitions $[0,1/22]$ at the ten rational points
\[
  0,\;\tfrac{1}{100},\;\tfrac{1}{50},\;\tfrac{3}{100},\;\tfrac{7}{200},\;
  \tfrac{1}{25},\;\tfrac{21}{500},\;\tfrac{11}{250},\;\tfrac{9}{200},\;
  \tfrac{1}{22},
\]
and proves the single rational-power inequality $A(r)^2>B(l)$ on each
of the nine cells $[l,r]$, each such inequality being a finite chain of
integer comparisons as above. Together these prove \eqref{eq:master} on
all of $(0,1/22]$. This argument uses no derivatives, no logarithms,
and no analysis beyond the monotonicity of $x^p$ in $p$.

\subsection{The derivative argument}\label{sec:derivative}

A second, independently implemented certificate proves the same
interval by a different route: it encloses
$D(p)=A(p)^2-B(p)$ and its derivative in rational intervals (the
logarithms of the seven rational bases are enclosed by an exact
$\operatorname{artanh}$ series with a controlled geometric tail),
proves $D'(p)<0$ throughout $[0,1/22]$, and proves $D(1/22)>0$. A
decreasing function that is positive at the right endpoint is positive
on the whole interval. The two arguments share the primitive facts of
Sections~\ref{sec:floorcap} and~\ref{sec:joint} but nothing at the
interval-reasoning layer.

\begin{remark}
The endpoint $1/22\approx 0.045$ is close to the intrinsic limit of the
present envelope: with the certified floor $\ell=301/100000$ and cap
$u=2/3$, the comparison $A(p)^2>B(p)$ itself fails shortly after
$p\approx 0.047$ (a floating-point observation offered as context, not
a certified statement). Extending the certified interval materially toward
the $\approx 0.1$ scale of CHLMW's reported numerics (observed in
different normalizations of the same supports; see
Section~\ref{sec:channels}) requires a larger certified floor (or a
genuinely sharper single-channel bound), not merely a finer
staircase.
\end{remark}

\section{Verification and provenance}\label{sec:verification}

The artifact ships four verification programs (all produced with AI
assistance under the author's direction; independence below means
separately written implementations and input transcriptions, not
separate authorship). Program by program:
\begin{itemize}
\item The \emph{interval-arithmetic program} reconstructs the
subspaces from its own transcription of CHLMW's printed bases,
encloses every algebraic entry in rational intervals with outward
rounding, converts entrywise error into an operator-norm bound, and
certifies the eight positivity statements of
Section~\ref{sec:floorcap} (the floor $301/100000$ and the cap $2/3$,
for both channels), the exact rational joint output matrix with its
characteristic polynomial, and the strict endpoint inequality
$A(1/22)^2>B(1/22)$.
\item The \emph{exact number-field program} was written separately,
from the paper's printed vectors and the certificate files alone. It
represents every matrix entry exactly in the number field
$\mathbb{Q}(i,\sqrt2,\sqrt3,\sqrt5)$ (as $16$ rational components),
decides positive semidefiniteness by exact Hermitian $LDL^{\dagger}$
factorization with exact sign decisions, recomputes the joint output
matrix and rejects on any mismatch with the shipped one, and certifies
the eigenvalue floor in the historical convention of the certificate
file ($\delta=1/400$ at order $p=1/100$). It does not check the cap or
the endpoint inequality; within its scope it is a second, independent
confirmation of the floor decompositions and the joint matrix.
\item The \emph{staircase program} certifies $A(p)^2>B(p)$ for every
$0<p\le 1/22$ from the ten rational grid points of
Section~\ref{sec:staircase}, reusing only three proven quantities
(floor, cap, joint spectrum) from an interval-arithmetic report that
it first re-binds to the shipped certificate files by hash.
\item The \emph{derivative program} certifies the same full-interval
claim by the independent argument of Section~\ref{sec:derivative} ---
a negative rational enclosure of the derivative on $[0,1/22]$ plus a
strictly positive endpoint gap --- under the same hash binding. Its
endpoint gap is a second in-artifact certification of the endpoint
inequality.
\end{itemize}
Thus the floor decompositions and the joint matrix are certified by
two independently written implementations; the endpoint inequality and
the full-interval claim are each certified by two programs
implementing independent arguments; the cap decompositions are
certified by the interval-arithmetic program (and rest on the same
exact witness data as the floor). Every program rejects deliberately
corrupted variants of its inputs with a nonzero exit code --- the
interval-arithmetic, staircase, and derivative programs through
built-in corruption controls, the exact number-field program under
documented one-value mutations of its two input files (a corrupted
witness entry, a corrupted floor parameter, a corrupted joint matrix)
--- a standard check that the tests are able to fail.

The witness matrices were found by an AI-assisted numerical search; the
search plays no role in the proof. Everything rests on the certificate
files --- a few small rational matrices and one rational spectrum ---
and on the four verification programs, which any reader can inspect and
run in a few seconds on a laptop. The certificate and verifier files,
with their SHA-256 hashes, are listed in Appendix~\ref{app:artifact}.

On the certificate side, the decomposable floor and cap witnesses of
Section~\ref{sec:floorcap} are explicit rational dual-feasible points
for the PPT relaxation --- the hierarchy of Doherty, Parrilo, and
Spedalieri~\cite{dps}, which CHLMW describe invoking --- not a tighter
relaxation. The step from a floating-point solver output to an exactly
checkable rational object has direct precedents in rigorous
semidefinite-programming postprocessing~\cite{jansson}, rational
sum-of-squares decompositions~\cite{peyrl}, and, in quantum
information, the recent rational certificates for quantum-code SDP
bounds of Angl\`es Munn\'e and Huber~\cite{anglesmunne} and the
certified rational bounds for quantum-theory optimization problems of
Naceur, Wang, Magron, and Ac\'in~\cite{naceur}.

This note does \emph{not} establish: an improved endpoint over CHLMW's
unprinted semidefinite-programming calculation, had its value been
recorded (they assert a bound ``of the order $10^{-2}$'' without
printing it; $1/22\approx 0.045$ is a convenient unit-fraction cutoff
with slack, and no comparison with their unprinted value is claimed in
either direction); optimality of the floor, cap, envelope, or endpoint; the true
minimum output entropies or minimizing inputs of either channel; or any
violation in the middle-order regime $[1/4,3/4]$, which
remains open~\cite{llw}.

\section*{Acknowledgments and contribution statement}

The explicit subspaces and the proof strategy of a uniform output
eigenvalue floor are due to Cubitt, Harrow, Leung, Montanaro, and
Winter~\cite{chlmw}, whose authors were contacted with these constants
in July 2026. AI systems were used for the numerical witness search,
for drafting text under the author's direction, and for implementing
the verification programs; all mathematical claims rest on the exact
certificates and the two independent verifications described in
Section~\ref{sec:verification}.

\appendix

\section{The printed subspaces}\label{app:subspaces}

For completeness we reproduce the CHLMW bases~\cite{chlmw} in the form
consumed by the verifiers. Write $\omega=e^{2\pi i/3}$. A vector of
$\mathbb{C}^4\otimes\mathbb{C}^3$ is displayed as a $3\times4$ array
$M$ with $|v\rangle=\sum_{b,a}M_{ba}\,|a\rangle_A|b\rangle_B$
(rows $b=1,2,3$ index the output system $B$, columns $a=1,\dots,4$ the
input system $A$). The six unnormalized basis vectors of $R$ are
\begin{gather*}
\begin{pmatrix}0&0&0&0\\1&0&0&0\\0&1&0&0\end{pmatrix},\qquad
\begin{pmatrix}1&0&0&0\\0&1&0&0\\0&0&1&0\end{pmatrix},\\[1ex]
\begin{pmatrix}1&0&0&0\\0&\omega&0&0\\0&0&\omega^2&0\end{pmatrix},\qquad
\begin{pmatrix}0&1&0&0\\0&0&\omega^2&0\\0&0&0&\omega\end{pmatrix},\\[1ex]
\begin{pmatrix}0&0&1&0\\0&0&0&-1\\0&0&0&0\end{pmatrix},\qquad
\begin{pmatrix}0&0&0&1\\0&0&0&0\\-1&0&0&0\end{pmatrix},
\end{gather*}
and the six unnormalized basis vectors of $S$ are
\begin{gather*}
\begin{pmatrix}0&0&0&0\\1&0&0&0\\0&-1&0&0\end{pmatrix},\qquad
\begin{pmatrix}1&0&0&0\\0&\omega^2&0&0\\0&0&\omega&0\end{pmatrix},\\[1ex]
\begin{pmatrix}0&1&0&0\\0&0&1&0\\0&0&0&1\end{pmatrix},\qquad
\begin{pmatrix}0&1&0&0\\0&0&\omega&0\\0&0&0&\omega^2\end{pmatrix},\\[1ex]
\begin{pmatrix}0&0&1&0\\0&0&0&1\\0&0&0&0\end{pmatrix},\qquad
\begin{pmatrix}0&0&0&1\\0&0&0&0\\1&0&0&0\end{pmatrix}.
\end{gather*}
The twelve vectors are pairwise orthogonal with squared norms $2$ or
$3$; $\bar S$ is obtained by conjugating every entry
($\omega\leftrightarrow\omega^2$). The marginals of
Section~\ref{sec:channels} follow by direct computation.

\section{Artifact and verification layout}\label{app:artifact}

The certificate data and all four verification programs are
distributed with this paper as a directory \texttt{artifact/}
containing \texttt{certificates/} (five JSON files: the floor
witnesses, the cap witnesses, the joint output matrix with its
spectrum, and the two full-interval certificates), \texttt{programs/}
(the four verification programs of Section~\ref{sec:verification},
standard-library Python only), a \texttt{MANIFEST.txt} listing the
hashes below, and a \texttt{README.md} giving one command per claim
together with the expected output. The exact bytes are identified by their SHA-256
hashes (each hash printed over two lines):

\begin{center}
\footnotesize\ttfamily
\begin{tabular}{@{}l@{}}
certificates/floor-witnesses.json\\
\hspace{1em}980b1af4df2f4984d8b8fce7822a6707\\
\hspace{1em}3cbd3bb3b45e8426c97c08a51a9de419\\[2pt]
certificates/cap-witnesses.json\\
\hspace{1em}5a8a5ec0d91cf6dcdd0ee9805ce25508\\
\hspace{1em}e36d6ae601341d6904e1b2d07c10e9e0\\[2pt]
certificates/joint-spectrum.json\\
\hspace{1em}a64ad02561a509422756859120ca2be0\\
\hspace{1em}4224ce310266aee352114982588ebf9d\\[2pt]
certificates/staircase-certificate.json\\
\hspace{1em}4ca802b3fc647a6f3168c99db4a27aa1\\
\hspace{1em}8f835cfed840d00156846ab04e85c821\\[2pt]
certificates/derivative-certificate.json\\
\hspace{1em}320a0308e150d8a8b5783e9122103c01\\
\hspace{1em}8998f718f7a7dd44b6f9cc21f7192c09\\[2pt]
programs/verify\_q4\_primitives.py\\
\hspace{1em}3bee50b95a8a34e851fb48f8d6c73ca6\\
\hspace{1em}ceea4e15aa303667e20a8fb4de0260fe\\[2pt]
programs/verify\_q4\_p001.py\\
\hspace{1em}e22ed9e0aaf648997c67f1356c913ab6\\
\hspace{1em}5ad9570b24b241240af96979332a099c\\[2pt]
programs/verify\_q4\_full\_interval.py\\
\hspace{1em}16ad4151d5c2e648a98a2f3da0c0733b\\
\hspace{1em}3497130a65adf71d4a0528c3d4626703\\[2pt]
programs/verify\_q4\_derivative.py\\
\hspace{1em}72fb8d612e814ad9c8133b9b83c706d7\\
\hspace{1em}8344a3b1616dab1ab75fc719aaa2034e
\end{tabular}
\end{center}

The artifact is archived at
DOI~\href{https://doi.org/10.5281/zenodo.21968558}{10.5281/zenodo.21968558};
this note is available as
\href{https://arxiv.org/abs/2608.17376}{arXiv:2608.17376}.

\end{document}